\documentclass[preprint]{elsarticle}
\usepackage{hyperref}
\usepackage{amsfonts}
\usepackage{graphicx}
\usepackage{float}
\usepackage{comment}

\usepackage{amsthm,amssymb}
\usepackage{amsmath}
\usepackage{xcolor}
\newcommand{\R}{\mathbb{R}}
\newcommand{\N}{\mathbb{N}}
\newcommand{\<}{\leqslant}
\renewcommand{\>}{\geqslant}
\newcommand{\Int}{\int\limits}
\newcommand{\Sum}{\sum\limits}
\newcommand{\lmin}{\lambda_{\textup{min}}}

\newcommand{\PC}{PC^1\bigl([-h,0],\R^n\bigr)}
\newcommand{\C}{C_{\infty}\bigl([-h,0],\R^n\bigr)}

\newcommand{\ph}{\varphi}

\newcommand{\dd}{\mathrm{d}}

\newcommand{\eps}{\varepsilon}

\newtheorem{theorem}{Theorem}
\newtheorem{lemma}{Lemma}
\newtheorem{proposition}{Proposition}
\newtheorem{remark}{Remark}
\newtheorem{corollary}{Corollary}
\newtheorem{definition}{Definition}

\newtheorem{assumption}{Assumption}
\newtheorem{property}{Property}

\journal{Journal of the Franklin Institute}

\begin{document}

\begin{frontmatter}

\title{Necessary and sufficient stability conditions for neutral-type delay systems: Polynomial approximations}


\author[1]{Gerson Portilla\corref{cor1}%
\fnref{fn1}}
\ead{gportilla@ctrl.cinvestav.mx}
\author[2]{Mathieu Bajodek}
\ead{mathieu.bajodek@cpe.fr}
\author[1]{Sabine Mondié\fnref{fn1}}
\ead{smondie@ctrl.cinvestav.mx}

\cortext[cor1]{Corresponding author}
\fntext[fn1]{The work of the first and third authors was supported by project CONACYT A1-S-24796, Mexico.}
\address[1]{Department of Automatic Control,
        CINVESTAV-IPN, 07360 Mexico D.F., Mexico}
\address[2]{Laboratoire d'Automatique, de Génie des Procédés et de Génie Pharmaceutique (LAGEPP), CPE Lyon, UCBL, Campus Lyontech-La Doua, 69100 Villeurbanne, France}

\begin{abstract}
A new necessary and sufficient stability test in a tractable number of operations for linear neutral-type delay systems is introduced. It is developed in the Lyapunov-Krasovskii framework via functionals with prescribed derivative. The necessary conditions, which stem from substituting any polynomial approximation of the functional argument, reduce to a quadratic form of monomials whose matrix is independent of the coefficients of the approximation under consideration. In the particular case of Chebyshev polynomials, the functional approximation error is quantified, leading to an estimate of the order of approximation such that the positive semi-definiteness of the functional is verified. Some examples illustrate the obtained results.
\end{abstract}


\end{frontmatter}


\section{Introduction}
In the Lyapunov-Krasovskii framework, stability theorems for neutral-type systems were introduced by Krasovskii in \cite{krasovskii1965}. In the monograph by Kharitonov \cite{kharitonov2013time}, functionals with prescribed derivative for linear neutral time-delay systems are constructed. They rely on the so-called delay Lyapunov matrix that, in analogy with delay-free linear systems, is obtained as the solution of a Lyapunov equation, which is now a set of three properties called symmetry, dynamic and algebraic. These functionals with prescribed derivative have been exploited to present sufficient \cite{alexandrova2019stability}, necessary and sufficient stability \cite{gomez2021necessary,portilla2023} results for neutral type delay systems. The above-mentioned results allowed the presentation of stability and instability theorems on a special set of bounded functional arguments as in \cite{alexandrova2019stability}, which turned out to be crucial for achieving tractable sufficiency results. 

Necessary conditions follow from the substitution of the functional argument. In particular, substituting the functional argument by an approximation in terms of the system's fundamental matrix reveals an elegant necessary condition in terms of the delay Lyapunov matrix \cite{gomez2021necessary}. 

For assessing the sufficiency of these conditions, one must find the order of approximation of an arbitrary initial function that satisfies the positivity condition of the functional for any initial approximated function. The fundamental matrix approximation is poor, thus it leads to very large approximation orders guaranteeing sufficiency. To remedy this issue, authors have proposed, to the cost of losing the delay Lyapunov matrix formulated criteria, better approximations of the functional argument, for example, the piece-wise linear approximations \cite{irina2022criterion} and Legendre polynomials in \cite{bajodek2022necessary} proposed for the retarded type case. The orthogonality, speed of convergence, and accuracy of Legendre polynomial approximations allowed achieving similar or tighter approximation orders piece-wise linear approximations \cite{irina2022criterion}.  In these contributions, after calculating the delay Lyapunov matrix, its derivative, and substitution of the approximated argument, the computation of integrals leads to an approximation of the functional, which, combined with sufficiency stability theorems, gives tractable conditions in the form of linear matrix inequalities \cite{fridman2014introduction,seuret2015hierarchy}.

Here, considering that it is well-known that polynomial approximations perform well and that any functional argument allows presenting necessary stability conditions for linear neutral time-delay systems, we propose an approach encompassing all classes of polynomial approximations. The result relies on the fact that any polynomial approximation can be reorganized as the product of a monomial vector multiplied by a vector of appropriate coefficients. The positive semi-definiteness of the resulting quadratic form delivers a necessary stability test for general polynomial approximation-based stability results. It is characterized by integrals of the delay Lyapunov matrix multiplied by monomials, which can be computed by a recursive method following the one introduced in \cite{irina2022criterion}.

For sufficiency, as the argument approximation order influence on the functional error must be quantified, one must restrict the analysis to a particular approximation method. Among polynomial approximations~\cite{phillips2003interpolation}, Taylor expansion, Bessel series, Lagrange interpolations, Padé approximants, and orthogonal polynomial approximations, such as  Chebyshev and Legendre, are available. We choose Chebyshev polynomials because, while they share the orthogonal property with Legendre polynomials, they outperform them with the super-geometric convergence of their approximation coefficients, ensuring a fast error convergence rate and simplification of the computations \cite{boyd2001chebyshev}. Some examples allow evidencing the reduced approximation orders, which are compared with previous results in the literature for neutral-type systems.

This paper is organized as follows. In Section~\ref{sec:preliminaries}, some preliminaries on linear neutral time-delay systems, functionals with prescribed derivative, and an instrumental instability result are introduced. Also, a monomial representation for any polynomial approximation and convergence properties of Chebyshev polynomials are presented.  In Section~\ref{sec:main_results}, a new necessary stability test resulting from general polynomial approximations and a stability criterion in a tractable number of operations is obtained for the case of Chebyshev polynomials. At the end, the technical issue of computing the integrals in terms of the delay Lyapunov matrix is tackled via a new recursive method described in Section~\ref{sec:rec_method} and some examples in Section~\ref{sec:examples} validate the results.

\textbf{Notation: } We consider the spaces of $\R^n$-valued piecewise continuously differentiable and smooth functions on $[-h,0]$, which are denoted by $\PC$ and $\C,$ respectively. They are equipped with the uniform norm 
\begin{equation*}
     \|\varphi\|_h=\sup_{\theta\in[-h,0]}\|\varphi(\theta)\|,
\end{equation*}
where $\|\cdot\|$ stands for the Euclidean norm for vectors and the spectral norm for matrices. $\Re(s)$ denotes the real part of a complex value $s$; $\lmin(W)$ is the smallest eigenvalue of a square matrix $W$; notation $k=\overline{n_1,n_2}$, where $n_1,~n_2\in \mathbb{Z}$, $n_1<n_2$, means that $k$ is an integer between $n_1$ and $n_2$; $I_n$ stands for the $n\times n$ identity matrix; $\mathrm{He}(M)$ means $M+M^\top$;  $\lceil \cdot \rceil$ denotes the ceiling function; $\{A_{ij}\}_{i,j=1}^r$ denotes a square block matrix, where $A_{ij}\in\R^{n\times n},~i,j=\overline{1,r}$, is the block in the $i$-th row and the $j$-th column; $\textup{vec}(X)$ means the column vectorization of a matrix $X$; Following \cite{kharitonov2013time}, $A\otimes B$ stands for the Kronecker product, namely,
\begin{equation*}
    A\otimes B = \begin{pmatrix}
    b_{11}A & b_{21}A & \cdots & b_{n1}A\\
    b_{12}A & b_{22}A & \cdots & b_{n2}A\\
    \vdots &\vdots & \ddots & \vdots\\
    b_{1n}A & b_{2n}A & \cdots & b_{nn}A\\
    \end{pmatrix},
\end{equation*}
where $B=\{b_{ij}\}_{i,j=1}^n,~b_{ij}\in\R$. Notice that, with this definition, $\textup{vec}(AXB)=(A\otimes B)\textup{vec}(X)$. For a symmetric matrix $\Lambda,$ the notation $\Lambda>0$ $(\Lambda\geq0)$ means that $\Lambda$ is a positive definite (positive semidefinite) matrix. $\mathcal{W}(z)$ denotes the Lambert function given by $\mathcal{W}:z \to y,~z\in \mathbb{R}_+,~y\in \mathbb{R}_+$, which is uniquely defined by the relation $y\textup{e}^y=z$.

\section{Preliminaries}\label{sec:preliminaries}
\subsection{Neutral-type delay systems and Lyapunov-Krasovskii functional}
Consider a neutral-type linear time-delay system of the form
\begin{equation} \label{eq:time-delay_sys}
    \frac{\mathrm{d}}{\mathrm{d}t}[x(t)-Dx(t-h)]=A_0x(t)+A_1x(t-h),
\end{equation}
where $h>0$, $A_0$, $A_1$, and $D$ are given real $n\times n$ matrices.  For $\ph\in\PC$, the solution $x(t)=x(t,\varphi)$ is a piecewise continuous function that satisfies system \eqref{eq:time-delay_sys} almost everywhere for $t\>0$, and the difference $x(t)-Dx(t-h)$ is continuous for $t\>0,$ except for possibly a countable number of points. The restriction of the solution $x(t,\ph)$ to the interval $[t-h,t],~t\>0$, is denoted by 
\begin{equation*}
    x_t(\ph):\theta \mapsto x(t+\theta,\ph),~\theta\in [-h,0].
\end{equation*}
\begin{definition}
System \eqref{eq:time-delay_sys} is exponentially stable if there exist $\gamma>0$ and $\sigma>0$ such that for any initial function $\ph\in\PC$,
\begin{equation*}
    \|x(t,\ph)\|\<\gamma e^{-\sigma t}\|\ph\|_h,~t\>0.
\end{equation*}
\end{definition}
We also introduce the following assumption on matrix $D$, which is a necessary condition for the exponential stability of system \eqref{eq:time-delay_sys}.
\begin{assumption}\label{as:schur}
Matrix $D$ is a Schur stable matrix, it satisfies $\|D\|<1$.
\end{assumption}

The delay Lyapunov matrix definition and a stability result on a particular set of initial functions in the Lyapunov-Krasovskii framework are recalled. 
\begin{definition}\label{def:matrix_U}\cite{kharitonov2013time}
Let $W\in\R^{n\times n}$ be a positive definite matrix. The \emph{delay Lyapunov matrix}~$U\!:[-h,h]\rightarrow \R^{n\times n}$ is a continuous matrix function, which satisfies the following properties.
\begin{enumerate}
    \item Dynamic property
          \begin{equation}\label{eq:dynamic_property}
              U'(\theta)-U'(\theta-h)D=U(\theta)A_0+U(\theta-h)A_1,~ \theta\in(0,h).
          \end{equation}
    \item Symmetry property 
          \begin{equation}\label{eq:symmetry_property}
              U^\top(\theta)=U(-\theta),~ \theta\in[-h,h].
          \end{equation}
    \item Algebraic property 
          \begin{equation}\label{eq:algebraic_property}
              P-D^\top PD=-W,
          \end{equation}
          where $P=\lim\limits_{\theta\rightarrow 0}\left(U'(+\theta)-U'(-\theta)\right)$.
\end{enumerate}
\end{definition}
Observe that, for negative values of the argument, the delay Lyapunov matrix satisfies
\begin{equation}\label{eq:dynamic_pro_tau_neg}
    U'(\theta)-D^{\top}U'(\theta+h)=-A_0^{\top}U(\theta)-A_1^{\top}U(\theta+h),~ \theta\in(-h,0).
\end{equation}
The following result gives conditions for the existence and uniqueness of the delay Lyapunov matrix. Notice that it does not rely on the stability or instability of system \eqref{eq:time-delay_sys}.
\begin{theorem}\label{th:lyap_condition}\cite{kharitonov2013time}
System \eqref{eq:time-delay_sys} admits a unique Lyapunov matrix if and only if the system satisfies the Lyapunov condition, i.e., if there exists $\eps>0$ such that any two points $s_1$ and $s_2$ of the spectrum of system \eqref{eq:time-delay_sys},
\begin{equation*}
    \Lambda=\Big\{s~\Big|~\textup{det}\Big(sI-A_0-se^{-sh} D- e^{-sh} A_1 \Big)=0\Big\},
\end{equation*}
satisfy $|s_1+s_2|>\eps$.
\end{theorem}
In the sequel, we use the following Lyapunov matrix-based functional \cite[Chapter~6, p.~242]{kharitonov2013time}
\begin{gather}\label{eq:functional}
    v_0(\ph) = (\ph(0) - D\ph(-h))^\top U(0) (\ph(0) - D\ph(-h)) + \Sum_{j=1}^6 I_j,
\end{gather}
where
\begin{align*}
    I_1 &= 2(\ph(0) -D\ph(-h))^\top\int_{-h}^0 U^{\top}(h+\theta)A_1\ph(\theta)\dd\theta\\ I_2&=- 2(\ph(0) -D\ph(-h))^\top\int_{-h}^0 U'^{\top}(h+\theta)D\ph(\theta)\dd\theta,\\
    I_3 &= \int_{-h}^0\int_{-h}^0 \ph^\top(\theta_1)A_1^{\top}U(\theta_1-\theta_2)A_1\ph(\theta_2)\dd\theta_2\dd\theta_1,\\
    I_4 &= 2\int_{-h}^0\int_{-h}^0 \ph^\top(\theta_1)A_1^{\top}U'(\theta_1-\theta_2)D\ph(\theta_2)\dd\theta_2\dd\theta_1,
\end{align*}
\begin{align*}
    I_5 &= -\int_{-h}^0\int_{-h}^{\theta_1} \ph^\top(\theta_1)D^{\top} U''(\theta_1-\theta_2)D\ph(\theta_2)\dd\theta_2\dd\theta_1\\
    &-\int_{-h}^0\int_{\theta_1}^{0} \ph^\top(\theta_1)D^{\top} U''(\theta_1-\theta_2)D\ph(\theta_2)\dd\theta_2\dd\theta_1,\\
    I_6 &= -\int_{-h}^0 \ph^\top(\theta)D^{\top}PD\ph(\theta)\dd\theta.
\end{align*}
and where $U(\theta)$ and $P$ are defined by Definition~\ref{def:matrix_U} with matrix $W>0$. Its time derivative along the solution of system \eqref{eq:time-delay_sys} has a prescribed quadratic negative value of the form
\begin{equation}\label{eq:presc_deri}
    \frac{\mathrm{d}v_0(x_t)}{\mathrm{d}t}=-x^\top(t)Wx(t),\quad t\geq 0.
\end{equation}
On the one hand, it was shown in \cite{alexandrova2019stability} that if system \eqref{eq:time-delay_sys} is exponentially stable, then functional $v_0(\ph)$ admits a quadratic lower bound considering a particular set of functions given by
\begin{equation}\label{set_S}
    \mathcal{S} = \Bigl\{\varphi\in \C\Bigl\arrowvert \|\varphi(0)\|=1, \|\varphi^{(k)}\|_h\leq r^k, \; \forall k\in \N\Bigr\},
\end{equation}
where $r=\frac{\|A_0\|+\|A_1\|}{1-\|D\|}$. 
On the other hand, this functional allows proving a crucial instability result reminded below.
\begin{lemma}\cite{alexandrova2019stability}\label{th:uns_v0}
 If system \eqref{eq:time-delay_sys} is unstable, then there exists a function $\ph\in \mathcal{S}$ such that
 \begin{equation*}
     v_0(\ph)<-a_0,\quad a_0=\dfrac{\lmin(W)}{4r}.
 \end{equation*}
\end{lemma}
\begin{remark}
   We can take $\Re(s)\<|s|\<r$ when the system is unstable, i.e., for any eigenvalue such that $\Re(s)>0.$
\end{remark}

\subsection{Preliminaries on polynomial approximation}\label{sec:pol_proj}
Next, we present some results on polynomial approximation. In particular, their monomial representation and the convergence of Chebyshev polynomial projections. 
\subsubsection{General polynomial approximations}
This section introduces a monomial representation of polynomial approximations for the argument $\varphi\in\PC$ of functional \eqref{eq:functional}. This representation reorganizes and expresses any polynomial projection as a decoupling of monomials and coefficients. Replacing this decoupling into functional \eqref{eq:functional} leads to a quadratic form, whose matrix is determined by monomials and the delay Lyapunov matrix (see Section~\ref{sec:nec_cond}). 

Let us consider a function $\ph\in\PC$ and a $N$-order polynomial approximation of this function given by
\begin{equation}\label{eq:ini_aprox}
    \ph_N(\theta)=P_N^\top(\theta) \mathcal{Q},\quad \theta\in  [-h,0],\quad N\in\N,
\end{equation}
$$P_N^\top(\theta) = \begin{bmatrix}p_0(\theta) I_n&p_1(\theta) I_n&\cdots&p_{N-1}(\theta)I_n\end{bmatrix},$$
where $\{p_k\}_{k\in\overline{0,N-1}}$ is a set of polynomials of order lower than $N$ and, vector $\mathcal{Q}\in\R^{nN}$ collects the coefficients associated to the chosen polynomial approximation. A monomial representation of the previous polynomial approximation can be expanded, resulting in the following expression: 
\begin{equation}\label{eq:P_N}
    \ph_N(\theta)=\Theta_N^{\top}(\theta)\Phi_N ,
\end{equation}
$$\Theta_N^\top(\theta) = \begin{bmatrix}I_n&\theta I_n&\cdots&\theta^{N-1}I_n\end{bmatrix},$$
where $\Theta_N(\theta)$ is a matrix of monomials, and $\Phi_N\in\R^{nN}$ is a vector with suitable coefficients after expansion. Then, it follows from \eqref{eq:P_N} that any function $\ph\in\PC$ can be written as 
\begin{equation}\label{eq:phi}
\begin{split}
    \ph(\theta)&=\ph_N(\theta)+\tilde{\ph}_N(\theta),\\
    \ph_N(\theta)&=\Theta_N^{\top}(\theta)\Phi_N,\quad   \theta\in [-h,0],
\end{split}    
\end{equation}
where the function $\ph_N(\theta)$, $\theta\in [-h,0]$, denotes the polynomial approximation of $\varphi\in\PC$ given by any polynomial approximation, and $\tilde{\ph}_N(\theta)$, $\theta\in [-h,0]$, stands for the corresponding approximation error. 
\begin{remark}
    For instance, this general framework gathers Lagrange interpolation, orthogonal projections, or Taylor series.
\end{remark}

\subsubsection{Orthogonal polynomial projections}\label{sec:orthogonal}
In the approximation theory, the orthogonal polynomial bases are characterized by convergence properties and accurate approximations with few terms. 
 In this section, we assume that the sequence $\{p_k\}_{k\in\N}$ in \eqref{eq:ini_aprox} is an orthogonal polynomial basis in the Banach space $\PC$ equipped with the norm $\ph\mapsto \Int_{-h}^0 \|\ph(\theta)\|^2w(\theta)\textup{d}\theta$, where $w$ is a weight function.

 We prove next a useful property of orthogonal polynomial projection on the sequence $\{p_k\}_{k\in\N}$.

\begin{property}
    For any function $\varphi\in\PC$ and any order $N\in\mathbb{N}$, its  orthogonal polynomial projection is given by
    \begin{equation}\label{eq:approx1}
        \varphi_N:
        \left\{
        \begin{aligned}
        \relax [&-h,0\relax ]\to\mathbb{R}^n,\\
        \theta&\mapsto \varphi_N(\theta):=P_N^\top(\theta) D_N^{-1} \int_{-h}^0 P_N(\tau)\varphi(\tau)w(\tau)\mathrm{d}\tau,
        \end{aligned}
        \right.
    \end{equation}
    where $D_N=\int_{-h}^0 P_N(\tau) P_N^\top(\tau)w(\tau)\mathrm{d}\tau$.
    This approximation can also be written as
    \begin{equation}\label{eq:approx2}
        \varphi_N(\theta) = \Theta_N^\top(\theta) S_N^{-1} \int_{-h}^0\Theta_N(\tau)\varphi(\tau)w(\tau)\mathrm{d}\tau,
    \end{equation}
    where $S_N=\int_{-h}^0\Theta_N(\tau)\Theta_N^\top(\tau)w(\tau)\mathrm{d}\tau.$ Here, $D_N$ and 
$S_N$ are symmetric Gramm-Schmidt orthonormalization matrices.
\end{property}
\begin{proof}
    The expression \eqref{eq:approx1} follows from the Galerkin-like methods introduced in \cite{fletcher1984computational} and stands as a definition. Thus, we proceed to prove \eqref{eq:approx2}. First, we express the monomials in terms of their orthogonal polynomials approximation at order $N$ for any $k\in\{0,\dots, N-1\}$ as follows
    \begin{equation*}
        \theta^k I_n = P_N^\top(\theta)  D_N^{-1}\int_{-h}^0 P_N(\tau)\tau^k w(\tau)\mathrm{d}\tau.
    \end{equation*}
    Since a polynomial approximation can reconstruct any polynomial of order $N$, this is indeed an equality. By concatenation, we get
    \begin{equation}\label{eq:phinelln}
        \Theta_N^\top(\theta) = P_N^\top(\theta)  D_N^{-1}\underbrace{\int_{-h}^0 P_N(\tau)\Theta_N^\top(\tau)w(\tau)\mathrm{d}\tau}_{=T_N},
    \end{equation}
    where $T_N\in\R^{nN\times nN}$ is a nonsingular triangular matrix that links the bases $\Theta_N$ and $P_N$. 
    Note also that $\left(D_N^{-1}\right)^\top=D_N^{-1}$ and that
    \begin{equation*}
        S_N = \int_{-h}^0\Theta_N(\tau)\Theta_N^\top(\tau)w(\tau)\mathrm{d}\tau = \underbrace{\int_{-h}^0\Theta_N(\tau) P_N^\top(\tau)w(\tau)\mathrm{d}\tau}_{T_N^{\top}}  D_N^{-1} T_N.
    \end{equation*}
    Gathering the two previous expressions yields 
    \begin{equation*}
    \begin{split}
        \Theta_N^\top(\theta) S_N^{-1} \int_{-h}^0\Theta_N(\tau)\varphi(\tau)w(\tau)\mathrm{d}\tau &= \left(P_N^\top(\theta)  D_N^{-1}T_N\right) \left(T_N^{-1}D_N (T_N^{\top})^{-1}\right)\\
        &\phantom{=}\times \left(T_N^\top D_N^{-1} \int_{-h}^0 P_N(\tau)\varphi(\tau)w(\tau)\mathrm{d}\tau\right)\\
        &= P_N^{\top}(\theta)D_N^{-1}\int_{-h}^0 P_N(\tau)\varphi(\tau)w(\tau)\mathrm{d}\tau.        
    \end{split}        
    \end{equation*}
    Therefore, we proved that the approximation forms provided in~\eqref{eq:approx1} and~\eqref{eq:approx2} are equivalent.
\end{proof}
Note that $\mathcal{Q}=D_N^{-1} \int_{-h}^0 P_N(\tau)\varphi(\tau)w(\tau)\mathrm{d}\tau$ and $\Phi_N=S_N^{-1} \int_{-h}^0\Theta_N(\tau)\varphi(\tau)w(\tau)\mathrm{d}\tau$ for the polynomial approximations \eqref{eq:ini_aprox} and \eqref{eq:P_N}, respectively. 
\begin{remark}
    The first expression~\eqref{eq:approx1} is a standard formulation that comes from orthogonal polynomial approximation theory \cite{boyd2001chebyshev}. Due to orthogonality of the sequence $\{p_k\}_{k\in\mathbb{N}}$, the Gramm-Schmidt orthonormalization matrix $D_N$ is diagonal which simplifies the maneuverability of the approximation structure. The second expression~\eqref{eq:approx2} is an the equivalent monomial representation and is associated to the Gramm-Schmidt orthonormalization matrix $S_N$. Such a structure encompasses Legendre or Chebyshev polynomial projections.
\end{remark}

\subsubsection{Chebyshev orthogonal polynomial projection}\label{sec:chebyshev_orthogonal}
Among orthogonal polynomial bases, Chebyshev and Legendre's polynomials are widely investigated and applied to approximation theory due to the super-geometric convergence of their coefficients \cite{scholl2022ode}. However, the Chebyshev coefficients decay a $\sqrt{N}$ factor faster than Legendre ones \cite{fox1968chebyshev,wang2012convergence}. Therefore, we choose the Chebyshev polynomials as a suitable polynomial approximation basis to tackle the estimation of the functional approximation error. In this context, $\{p_k\}_{k\in\mathbb{N}}$ are the Chebyshev polynomial, which are given by $p_0(\theta)=1,~p_1(\theta)=\frac{2\theta}{h}+1,~p_2=\frac{8\theta}{h} + \frac{8\theta^2}{h^2} + 1$ for low orders, and the weight is equal to $w(\theta)=\frac{1}{\sqrt{1-(-\frac{2\theta+h}{h})^2}}$ for $\theta\in[-h,0]$. 

Next, a lemma on the Chebyshev convergence rate on the set of functions $\mathcal S$ is stated. 
\begin{lemma}\label{lemma:estimate_ph_tilde}
Consider $\ph\in\mathcal{S}$. The Chebyshev approximation error $\tilde{\ph}_N$ satisfies the following inequality 
\begin{equation}\label{eq:estimate_ph_tilde}
  \|\tilde{\ph}_N\|_h\< \frac{4\left(\frac{hr}{2}\right)^{N}}{N!} .
\end{equation}
\end{lemma}   
\begin{proof}
The proof is based on Theorem~4.2 in \cite{trefethen2008gauss} with a slight modification. First of all, we start shifting the Chebyshev polynomials and coefficients on the interval $[-h,0]$, as follows 
    \begin{equation*}
        a_k=\frac{4}{\pi h}\Int_{-h}^0 \frac{\ph(\theta)(-1)^k p_k(-\frac{2\theta+h}{h})}{\sqrt{1-(-\frac{2\theta+h}{h})^2}}\dd \theta,\quad \theta\in[-h,0].
    \end{equation*}
    The change of variable $x=-\frac{2\theta+h}{h}$ yields
    \begin{equation*}
        a_k=\frac{2(-1)^k}{\pi}\Int_{-1}^{1} \frac{\ph(-\frac{h}{2}(1+x)) p_k(x)}{\sqrt{1-x^2}}\dd x,\quad x\in[-1,1].
    \end{equation*}
    From here, one proceeds as in Theorem 4.2 in \cite{trefethen2008gauss}. Thus, the coefficients admit the upper bound
    \begin{equation*}
        |a_k|\< \frac{2\left(\frac{hr}{2}\right)^{m+1}}{k(k-1)\cdots(k-m)},\quad m\> 1,\quad k\> m+1.
    \end{equation*}
    Then, we obtain the following estimate for the approximation error
    \begin{equation*}
    \begin{aligned}
        \lVert \tilde{\varphi}_N\rVert_h
        &\< \sum_{k=N}^\infty \frac{2\left(\frac{hr}{2}\right)^{m+1}}{k(k-1)\cdots(k-m)},\\
        &= \frac{2\left(\frac{hr}{2}\right)^{m+1}}{m} \sum_{k=N}^\infty \left(\tfrac{1}{(k-1)\cdots(k-m)}-\tfrac{1}{k(k-1)\cdots(k-m+1)}\right),\\
        &\< \frac{2\left(\frac{hr}{2}\right)^{m+1}}{m(N-1)\cdots(N-m)},\quad m\>1,\quad N\>m+1.
    \end{aligned}
    \end{equation*}    
    Finally, by taking $m=N-1$, we conclude that for all $N\> 1$, $\lVert \tilde{\varphi}_N\rVert_h$ is bounded by
    \begin{equation*}
        \lVert \tilde{\varphi}_N\rVert_h \< \frac{2\left(\frac{hr}{2}\right)^{N}}{N!}\frac{N}{N-1}\< \frac{4\left(\frac{hr}{2}\right)^{N}}{N!}.
    \end{equation*}
    For $N=0$, we have $\lVert\tilde{\varphi}_N\rVert_h\<2\lVert\varphi\rVert_h\<2$, which also satisfies \eqref{eq:estimate_ph_tilde}, thus \eqref{eq:estimate_ph_tilde} holds for all $N\in\N$.
\end{proof}

\begin{remark}
    For $\ph\in \mathcal{S}$, Chebyshev polynomial approximation error converges with a super-geometric convergence rate. Through orthogonal projections, it is the best approximation to minimize the norm of the approximation error \cite{boyd2001chebyshev,fox1968chebyshev}. 
\end{remark}

\section{Main results}\label{sec:main_results}
In this section, we first present necessary positive semi-definite conditions stemming from general polynomial approximations of the functional argument $\ph$. Second, for the Chebyshev orthogonal polynomials projection, a sufficiency positive semi-definite condition and a stability condition verifiable in a tractable number of operations are presented. 

\subsection{Necessary conditions: general polynomial approximations}\label{sec:nec_cond}
By substituting the polynomial approximation \eqref{eq:P_N} of $\ph\in\PC$ into the functional \eqref{eq:functional}, we obtain the functional $v_0(\ph_N)$, which is the polynomial approximation of functional \eqref{eq:functional}. Observe that $I_1$ and $I_2$ reduce to
\begin{multline}\label{eq:J1N_J2N}
    I_1(\ph_N)+I_2(\ph_N)=2(\ph(0) - D\ph(-h))^{\top}\\\times\int_{-h}^0 \left(U^{\top}(h+\theta)A_1-U'^{\top}(h+\theta)D\right)\Theta_N^{\top}(\theta)\Phi_N\dd\theta
    \\=2(\ph(0) - D\ph(-h))^{\top}(J_{1N}-J_{2N}) \Phi_N,    
\end{multline}
where
\begin{align*}
    J_{1N} = \int_{-h}^0 U^{\top}(h+\theta)A_1\Theta_N^\top(\theta)\dd\theta,\quad J_{2N} = \int_{-h}^0 U^{\prime\top}(h+\theta)D\Theta_N^\top(\theta)\dd\theta.
\end{align*}
Evaluating the remaining summands of the functional, we get
\begin{equation}\label{eq:J3N_J7N}
\begin{split}
   \Sum_{i=3}^{6} I_i(\ph_N)&=\int_{-h}^0\int_{-h}^0 \Phi_N^\top\Theta_N(\theta_1)\left(A_1^{\top}U(\theta_1-\theta_2)A_1\right.\\
   &\phantom{=}\left.+2A_1^{\top}U'(\theta_1-\theta_2)D\right)\Theta_N^\top(\theta_2)\Phi_N\dd\theta_2\dd\theta_1\\
    &\phantom{=}-\int_{-h}^0\int_{-h}^{\theta_1} \Phi_N^\top\Theta_N(\theta_1)D^\top U''(\theta_1-\theta_2)D\Theta_N^\top(\theta_2)\Phi_N\dd\theta_2\dd\theta_1\\
&\phantom{=}-\int_{-h}^0\int_{\theta_1}^{0} \Phi_N^\top\Theta_N(\theta_1)D^\top U''(\theta_1-\theta_2)D\Theta_N^\top(\theta_2)\Phi_N\dd\theta_2\dd\theta_1\\
&\phantom{=}-\int_{-h}^0\Phi_N^\top\Theta_N(\theta)D^\top P D\Theta_N^\top(\theta)\Phi_N\dd\theta\\
&=\Phi_N^\top (J_{3N}+2J_{4N}-J_{5N}-J_{6N}) \Phi_N.
\end{split}    
\end{equation}
Here,
\begin{align*}
    J_{3N} &= \int_{-h}^0\int_{-h}^0 \Theta_N(\theta_1)A_1^\top U(\theta_1-\theta_2)A_1\Theta_N^\top(\theta_2)\dd\theta_2\dd\theta_1,\\
    J_{4N} &= \int_{-h}^0\int_{-h}^0 \Theta_N(\theta_1)A_1^\top U^\prime(\theta_1-\theta_2)D\Theta_N^\top(\theta_2)\dd\theta_2\dd\theta_1,\\
    J_{5N} &= \int_{-h}^0 \left(\int_{-h}^{\theta_1} \Theta_N(\theta_1)D^\top U''(\theta_1-\theta_2)D\Theta_N^\top(\theta_2)\dd\theta_2\right.\\
    &\phantom{=}\left.+\int_{\theta_1}^{0} \Theta_N(\theta_1)D^\top U''(\theta_1-\theta_2)D\Theta_N^\top(\theta_2)\dd\theta_2\right)\dd\theta_1,\\
     J_{6N} &= \int_{-h}^0 \Theta_N(\theta)D^\top P D\Theta_N^\top(\theta)\dd\theta.\\
\end{align*}

Note that, for any $\Psi(\theta)=\Psi^\top(-\theta),~\theta>0$, we have
\begin{equation*}
    \Int_{-h}^0\Int_{\theta_1}^{0} \theta_1^i\theta_2^j\Psi(\theta_1-\theta_2)\dd\theta_2\dd\theta_1 
    =\Int_{-h}^0\Int_{-h}^{\theta_1} \theta_2^{i}\theta_1^{j}\Psi^\top(\theta_1-\theta_2) \dd \theta_2 \dd \theta_1.
\end{equation*}


Hence, $J_{3N},~J_{4N}$ and$~J_{5N}$ can be rewritten as
\begin{align*}
    J_{3N} &= \mathrm{He}\!\left(\!\Int_{-h}^0\Int_{-h}^{\theta_1} \Theta_N(\theta_1)A_1^\top U(\theta_1-\theta_2)A_1\Theta_N^\top(\theta_2)\dd\theta_2\dd\theta_1\!\right)\!,\\
    J_{4N} &= \mathrm{He}\!\left(\!\Int_{-h}^0\Int_{-h}^{\theta_1} \Theta_N(\theta_1)A_1^\top U^\prime(\theta_1-\theta_2)D\Theta_N^\top(\theta_2)\dd\theta_2\dd\theta_1\!\right)\!,\\
    J_{5N} &= \mathrm{He}\!\left(\!\Int_{-h}^0\Int_{-h}^{\theta_1} \Theta_N(\theta_1)D^\top U^{\prime\prime}(\theta_1-\theta_2)D\Theta_N^\top(\theta_2)\dd\theta_2\dd\theta_1\!\right)\!.
\end{align*}
Gathering the approximated summands, we rewrite the polynomial approximation of functional \eqref{eq:functional} as:
\begin{equation}\label{eq:v_phN}
   v_0(\ph_N)= \Xi^{\top}\mathbf{P}_N \Xi,
\end{equation}
where
\begin{equation*} 
    \mathbf{P}_N = \begin{bmatrix}
        J_{0N} & J_{1N}+J_{2N}\\
        \star & J_{3N}+2J_{4N}-J_{5N}-J_{6N}
    \end{bmatrix},\quad J_{0N}=U(0),
\end{equation*}
\begin{equation*}
    \Xi=[(\ph(0)-D\ph(-h))^\top\quad \Phi_N^\top]^\top.
\end{equation*}
Now, for \textit{any} polynomial approximation of $\ph$ and \textit{any} order of approximation $N$, we give a necessary stability condition through the functional approximation \eqref{eq:v_phN}. 
\begin{theorem}\label{th:necessary_condition}
If system \eqref{eq:time-delay_sys} is exponentially stable, then the matrix $\mathbf{P}_{N}\geq 0,~\forall N.$
\end{theorem}
\begin{proof}
Since system \eqref{eq:time-delay_sys} is exponentially stable, it is also stable, thus by construction of the functional, we have 
\begin{equation*}
    v_0(\ph)=\int_0^{+\infty} x^T(t,\ph)Wx(t,\ph)\textup{d}t\>0,
\end{equation*}
for any $\ph\in\PC$. In particular, for the polynomial approximation $\ph_N:[-h,0]\to\mathbb{R}^n$ given by \eqref{eq:P_N}, it follows from \eqref{eq:v_phN} that 
\begin{equation*}
    v_0(\ph_N)=\Xi^{\top}\mathbf{P}_N \Xi\>0,
\end{equation*}
for any $\Xi$ and $N$, which implies $\mathbf{P}_N\geq 0.$
\end{proof}
\begin{remark}
    Notice that an approximation of any order, even small, provides a necessary condition. 
\end{remark}

\subsection{Sufficient condition: Chebyshev polynomial projection}\label{sec:suff_cond}
Here, we focus on the Chebyshev polynomial projection introduced in Section~\ref{sec:chebyshev_orthogonal}. To conclude on the exponential stability of system \eqref{eq:time-delay_sys}, we first analyze how the approximation error of the Chebyshev orthogonal polynomial affects the functional approximation error $\Upsilon_N$. Next, we combine this result with the sufficient stability condition for functional $v(\ph)$ established in Lemma~\ref{th:uns_v0}.

The approximation error on the functional argument $\ph$ is conveyed to functional \eqref{eq:functional}, leading to the following expression for the functional approximation error
\begin{equation}\label{eq:functional_error}
    \Upsilon_N= v_0(\varphi)-v_0(\varphi_N),\quad \ph\in\mathcal{S}.
\end{equation}
The value $\Upsilon_N$ quantifies the error between $v_0(\varphi_N)$ and $v_0(\varphi)$ on the set of functions $\ph\in\mathcal{S}$. We now proceed to compute and estimate the functional approximation error in order to transfer the convergence properties of Chebyshev to the error~$\Upsilon_N$.

 Taking into account that if the Lyapunov condition in Theorem~\ref{th:lyap_condition} holds, then $U(s)$ is guaranteed to be bounded, we introduce the constants
\begin{equation*}
    M_1=\sup_{\theta\in (0,h)}\|U^{\top}(\theta)A_1-U'^{\top}(\theta)D\|,
\end{equation*}
\begin{equation*}
    M_2=\sup_{\theta\in (0,h)}\|A_1^{\top}U(\theta)A_1+2A_1^{\top}U'(\theta)D - D^\top U''(\theta)D\|,
\end{equation*}
\begin{equation*}
    M_3=\|D^\top P D\|.
\end{equation*}
Denote $\Psi(\theta)=A_1^{\top}U(\theta)A_1+2A_1^{\top}U'(\theta)D$ and notice that the term $\Upsilon_N$ can be presented as
\begin{equation*}    \Upsilon_N=\Sum_{i=1}^6 I_i^{\mathrm{error}},
\end{equation*}
where 
\begin{equation*}
    I_1^{\mathrm{error}}+I_2^{\mathrm{error}}=2[\varphi(0)-D\varphi(-h)]^\top \int_{-h}^0 \left(U^{\top}(h+\theta)A_1-U'^{\top}(h+\theta)D\right)\tilde{\ph}_N(\theta)\mathrm{d}\theta,   
\end{equation*}
\begin{equation*}
\begin{split}
   &\Sum_{i=3}^{6} I_i^{\mathrm{error}}=\int_{-h}^0\int_{-h}^0 \ph^\top(\theta_1)\Psi(\theta_1-\theta_2)\ph(\theta_2)\dd\theta_2\dd\theta_1\\
   &-\int_{-h}^0\int_{-h}^0 \ph^\top_N(\theta_1)\Psi(\theta_1-\theta_2)\ph_N(\theta_2)\dd\theta_2\dd\theta_1\\
    &-\int_{-h}^0\int_{-h}^{\theta_1} \ph^\top(\theta_1)D^\top U''(\theta_1-\theta_2)D\ph(\theta_2)\dd\theta_2\dd\theta_1\\
    &+\int_{-h}^0\int_{-h}^{\theta_1} \ph_N^\top(\theta_1)D^\top U''(\theta_1-\theta_2)D\ph_N(\theta_2)\dd\theta_2\dd\theta_1\\
&-\int_{-h}^0\int_{\theta_1}^{0} \ph^\top(\theta_1)D^\top U''(\theta_1-\theta_2)D\ph(\theta_2)\dd\theta_2\dd\theta_1
\end{split}    
\end{equation*}
\begin{equation*}
\begin{split}
&+\int_{-h}^0\int_{\theta_1}^{0} \ph_N^\top(\theta_1)D^\top U''(\theta_1-\theta_2)D\ph_N(\theta_2)\dd\theta_2\dd\theta_1\\
&-\int_{-h}^0\ph^\top(\theta)D^\top P D\ph(\theta)\dd\theta+\int_{-h}^0\ph_N^\top(\theta)D^\top P D\ph_N(\theta)\dd\theta.
\end{split}    
\end{equation*}
Substituting $\ph_N=\ph-\tilde{\ph}_N$, we have

\begin{equation*}
\begin{split}
   &\Sum_{i=3}^{6} I_i^{\mathrm{error}}= 2\int_{-h}^0\int_{-h}^0 \ph^{\top}(\theta_1)\Psi(\theta_1-\theta_2)\tilde{\ph}_N(\theta_2)\dd\theta_2\dd\theta_1\\
   &-\int_{-h}^0\int_{-h}^0 \tilde{\ph}_N^\top(\theta_1)\Psi(\theta_1-\theta_2)\tilde{\ph}_N(\theta_2)\dd\theta_2\dd\theta_1\\
    &-2\int_{-h}^0\int_{-h}^{\theta_1} \ph^{\top}(\theta_1)D^{\top} U''(\theta_1-\theta_2)D\tilde{\ph}_N(\theta_2)\dd\theta_2\dd\theta_1\\
    &+\int_{-h}^0\int_{-h}^{\theta_1} \tilde{\ph}_N^\top(\theta_1)D^{\top} U''(\theta_1-\theta_2)D\tilde{\ph}_N(\theta_2)\dd\theta_2\dd\theta_1\\
    &-2\int_{-h}^0\int_{\theta_1}^{0} \ph^{\top}(\theta_1)D^{\top} U''(\theta_1-\theta_2)D\tilde{\ph}_N(\theta_2)\dd\theta_2\dd\theta_1\\
    &+\int_{-h}^0\int_{\theta_1}^{0} \tilde{\ph}_N^\top(\theta_1)D^{\top} U''(\theta_1-\theta_2)D\tilde{\ph}_N(\theta_2)\dd\theta_2\dd\theta_1\\
&-2\int_{-h}^0\tilde{\ph}_N^\top(\theta)D^{\top} P D\ph(\theta)\dd\theta+\int_{-h}^0\tilde{\ph}_N^\top(\theta)D^{\top} P D\tilde{\ph}_N(\theta)\dd\theta.
\end{split}    
\end{equation*}

Considering $\varphi\in \mathcal{S}$, $\|\ph(\theta)\|\<\|\ph(0)\|=1$ and Lemma~\ref{lemma:estimate_ph_tilde}, the summands of $\Upsilon_N$ are estimated by:
\begin{equation*}
    |I_1^{\mathrm{error}}+I_2^{\mathrm{error}}|\< 2h(1+\|D\|)M_1\|\tilde{\ph}_N\|_h,  
\end{equation*}
\begin{equation*}
\left|\Sum_{i=3}^{6} I_i^{\mathrm{error}}\right|\< 2h(hM_2 + M_3)\|\tilde{\ph}_N\|_h+h(hM_2+M_3)\|\tilde{\ph}_N\|_h^2.
\end{equation*}
Hence, the functional approximation error $\Upsilon_N$ admits the upper bound
\begin{equation}\label{eq:est_functional_error}
    |\Upsilon_N|\< 2c_1\|\tilde{\ph}_N\|_h+c_2\|\tilde{\ph}_N\|_h^2,
\end{equation}  
$$c_1=h((1+\|D\|)M_1+hM_2+M_3),\quad c_2=h(hM_2+M_3).$$
Using the upper bound \eqref{eq:estimate_ph_tilde} on the argument error of the Chebyshev approximation of $\ph$, we obtain
\begin{equation}\label{eq:est_functional_error_2}
    \small |\Upsilon_N|\< 2c_1\frac{4\left(\frac{hr}{2}\right)^{N}}{N!}+c_2\left(\frac{4\left(\frac{hr}{2}\right)^{N}}{N!}\right)^2\<\delta_N=\frac{(8c_1\!+\!16c_2)\mu^N}{N!},
\end{equation} 
where $\mu=\max\left\{\frac{hr}{2},\left(\frac{hr}{2}\right)^2\right\}$ and $r$ is defined in \eqref{set_S}.  

Next, a sufficiency stability condition based on the functional approximation \eqref{eq:v_phN}.
\begin{theorem}\label{th:sufficiency condition}
Let $\mathbf{P}_{N}$ in \eqref{eq:v_phN} where $\Phi_N$ is determined through Chebyshev polynomial projection. If the Lyapunov condition holds and there exists a natural number $N$ such that
\begin{equation}\label{eq:suff_condition}
    \mathbf{P}_{N}-\mathcal{X}_N\geq 0,
\end{equation}
where the matrix $\mathcal{X}_N$ is of the same dimension as $\mathbf{P}_{N}$ with the upper left $n\times n$ block equal to $\frac{\delta_N}{(1-\|D\|)^2}I_n$ and all other blocks equal to zero, and $\delta_N$ defined in \eqref{eq:est_functional_error_2}, then system \eqref{eq:time-delay_sys} is exponentially stable.
\end{theorem}
\begin{proof}
    By contradiction, we assume that system \eqref{eq:time-delay_sys} is unstable, but condition \eqref{eq:suff_condition} holds. Now, it follows from \eqref{eq:v_phN}, \eqref{eq:functional_error} and \eqref{eq:est_functional_error_2} that functional \eqref{eq:functional} admits the lower bound
\begin{equation*}
\begin{split}
    v_0(\ph)&\>\Xi^{\top}\mathbf{P}_N \Xi-\delta_N \|\ph(0)\|^2\\
    &\>\Xi^{\top}\mathbf{P}_N \Xi-\frac{\delta_N}{(1-\|D\|)^2} \|\ph(0)-D\ph(-h)\|^2\\
    &=\Xi^{\top}(\mathbf{P}_{N}-\mathcal{X}_N) \Xi.
\end{split}   
\end{equation*}  
 According to the assumption of instability of the system and Lemma~\ref{th:uns_v0}, there exists a constant $a_0>0$ such that
\begin{equation*}
    \Xi^{\top}(\mathbf{P}_{N}-\mathcal{X}_N) \Xi\<v_0(\ph)<-a_0.
\end{equation*}
This implies that $\mathbf{P}_{N}-\mathcal{X}_N<0$, which is a contradiction, thus \eqref{eq:suff_condition} is proved.
\end{proof}

\begin{corollary}\label{cor:nec_and_suf}
System \eqref{eq:time-delay_sys} is exponentially stable, if and only if the Lyapunov condition holds and the matrix $\mathbf{P}_{N^{\star}}\geq 0$, where the order of approximation is given by
\begin{equation}\label{eq:nstar}
    N^{\star}=\left\lceil \mu \exp\left[1+\mathcal{W}\left(\frac{1}{\exp(1)\mu}\log \left(\frac{8c_1+16c_2}{a_0}\right)\right)\right]\right\rceil.
\end{equation}
Here, the constant $a_0$ is determined in Theorem~\ref{th:uns_v0}.
\end{corollary}
\begin{proof}
\textit{[Necessity]} The necessity follows from Theorem~\ref{th:necessary_condition}. In particular, $N=N^{\star}$.

\textit{[Sufficiency]} By contradiction, assume that system \eqref{eq:time-delay_sys} is unstable but that $\mathbf{P}_{N^{\star}}$ is positive semi-definite. Moreover, recall that $v_0(\ph)$ can be expressed as
\begin{equation*}
    v_0(\varphi)=v_0(\varphi_{N^{\star}})+\Upsilon_{N^{\star}}.
\end{equation*} 
By Lemma~\ref{th:uns_v0}, there exists $\varphi\in \mathcal{S}$ such that $v_0(\ph)<-a_0$, which implies that 
\begin{equation*}
    v_0(\varphi_{N^{\star}})=v_0(\varphi)-\Upsilon_{N^{\star}}< -a_0+|\Upsilon_{N^{\star}}|.
\end{equation*}
Then, for $N^\star$, it follows from \eqref{eq:est_functional_error} that the functional approximation error admits the upper bound
\begin{equation*}
        |\Upsilon_{N^\star}|\< \frac{(8c_1+16c_2)\mu^{N^\star}}{N^\star!}.
\end{equation*}
Applying the logarithm to \eqref{eq:est_functional_error} and the Maclaurin integral test give the following upper bound
\begin{equation*}
\begin{split}
    \log \left(|\Upsilon_{N^\star}|\right) &\< \log (8c_1+16c_2) + N^\star\log \left(\mu\right)  - \int_{0}^{N^\star}\log(s)\mathrm{d}s \\
    &= \log \left(8c_1+16c_2\right) - N^\star\log \left(\frac{N^\star}{\exp(1)\mu}\right),
\end{split}    
\end{equation*}   
equivalently,
\begin{equation*}
    \frac{1}{\exp(1)\mu} \log \left(\frac{|\Upsilon_{N^\star}|}{8c_1+16c_2}\right) \< - \frac{N^\star}{\exp(1)\mu}\log \left(\frac{N^\star}{\exp(1)\mu}\right).
\end{equation*} 
Under the condition
\begin{equation*}
   \log \left(\frac{N}{\exp(1)\mu}\right)\> \mathcal{W}\left(\frac{1}{\exp(1)\mu}\log \left(\frac{8c_1+16c_2}{a_0}\right)\right),
\end{equation*}
which is satisfied for order $N^\star$ given by~\eqref{eq:nstar}, we have $\log(|\Upsilon_{N^\star}|)\<\log(a_0)$ implying that $|\Upsilon_{N^{\star}}|\< a_0$. Hence, 
\begin{equation*}
    v_0(\varphi_{N^{\star}})< -a_0+|\Upsilon_{N^{\star}}|\<0,
\end{equation*}
contradicting the statement of positive semi-definiteness.
\end{proof}
 Corollary~\ref{cor:nec_and_suf} merges Theorem~\ref{th:necessary_condition} and \ref{th:sufficiency condition} to reach the following necessary and sufficient stability condition.

\section{Recursive method for projection computation}\label{sec:rec_method}
Following closely the ideas introduced in \cite{irina2022criterion}, we devote this section to present efficient tools for computing the integrals involved in determining matrix $\mathbf{P}_{N}$. These matrices are essential for conducting the stability test discussed in the preceding sections. The computation of these integrals imposes a significant numerical burden, potentially compromising the proper evaluation of the condition. Our aim is to ensure accurate condition testing by reducing the computational load.

Observe that $J_{1N},~J_{2N},~J_{3N},~J_{4N},~J_{5N}$ and~$J_{6N}$ defined in \eqref{eq:J1N_J2N} and~\eqref{eq:J3N_J7N} involve the following elementary integral matrices for $k=\overline{0,N-1}$ and $i,j=\overline{0,N-1}$:
\begin{align*}
&G_k=\Int_{-h}^0 U(h+\theta)\theta^k\dd s,\quad\Bar{G}_k=\Int_{-h}^0 U(\theta)\theta^k \dd \theta,\\
&H_{ij}= \Int_{-h}^0\Int_{-h}^{\theta_1} \theta_1^i\theta_2^jU(\theta_1-\theta_2) \dd \theta_2 \dd \theta_1,\quad
\Bar{H}_{ij}=\Int_{-h}^0\Int_{-h}^{\theta_1} \theta_1^i\theta_2^jU(\theta_1-\theta_2-h) \dd \theta_2 \dd \theta_1,
\end{align*}
which must be computed. To carry out this task, let us consider the dynamic property \eqref{eq:dynamic_property}, denote 
\begin{equation*}
    L=\begin{bmatrix}
      I\otimes I & -I\otimes D\\
      -D^{\top}\otimes I & I\otimes I
    \end{bmatrix}^{-1}\begin{bmatrix}
      I\otimes A_0 & I\otimes A_1\\
      - A_1^{\top}\otimes I & -A_0^{\top}\otimes I
    \end{bmatrix},
\end{equation*}
and define the vectors $G_k^v=\textup{vec}(G_k)$, $\bar{G}_k^v=\textup{vec}(\bar{G}_k)$, $H_{ij}^v=\textup{vec}(H_{ij})$, $\bar{H}_{ij}^v=\textup{vec}(\bar{H}_{ij})$, $U^v(0)=\textup{vec}(U(0))$ and $U^v(-h)=\textup{vec}(U(-h))$. 
Next, we introduce two propositions allowing the recursive computation of the above integral.
\begin{proposition}\label{pro:G}
    If the Lyapunov condition holds and $\mathrm{det}(L)\neq 0$, then the set of matrices $\{G_k\}_{k=\overline{0,N-1}}$ can be computed as
    \begin{align}
        \begin{bmatrix}G_0^v\\\Bar{G}_0^v\end{bmatrix} &= L^{-1}
    \begin{bmatrix}
      U^v(h)-U^v(0)\\
      U^v(0)-U^v(-h)
    \end{bmatrix},\label{eq:G0}\\
     \begin{bmatrix}G_k^v\\\Bar{G}_k^v\end{bmatrix} &= -L^{-1}\left(k\begin{bmatrix}G_{k-1}^v\\\Bar{G}_{k-1}^v\end{bmatrix} + (-h)^k \begin{bmatrix}U^v(0)\\U^v(-h)\end{bmatrix}\right),\, k\neq 0.
     \label{eq:Gk}
    \end{align}
\end{proposition}
\begin{proof}
     Integrating the dynamic property~\eqref{eq:dynamic_property} yields
    \begin{equation*}
        \Int_{-h}^0 U'(h+\theta)\dd \theta - \Int_{-h}^0 U'(\theta)\dd \theta D= \Int_{-h}^0 U(h+\theta) \dd \theta A_0 + \Int_{-h}^0 U(\theta) \dd \theta A_1.   
    \end{equation*}
    We also note that
    \begin{equation*}
        \int_{-h}^0U(h+\theta)\dd\theta=G_0,\quad \Bar{G}_0=\Int_{-h}^0 U(\theta) \dd \theta = \Int_{-h}^0 U^{\top}(-\theta) \dd \theta = \Int_0^{h} U^{\top}(\theta) \dd \theta = \Int_{-h}^0 U^{\top}(h+\theta) \dd \theta = G_0^\top.
    \end{equation*}
    Hence, we have
    \begin{equation*}
        U(h)-U(0)-U(0)D+U(-h)D = G_0 A_0 + G_0^\top A_1.
    \end{equation*}
    Notice that this equation can be rewritten as
    \begin{equation*}
        U(0)-U(-h)-D^\top U(h)D+D^\top U(0) = -A_0^\top G_0^\top  - A_1^\top G_0 .
    \end{equation*}    
    By vectorization based on Kronecker product properties \cite{roger1994topics}, we arrive at the following system of linear algebraic equations
    \begin{equation*}
        \begin{bmatrix}
          I\otimes A_0 & I\otimes A_1\\
          - A_1^{\top}\otimes I & -A_0^{\top}\otimes I
        \end{bmatrix}\begin{bmatrix}
         G_0^v\\ \bar{G}_0^v
        \end{bmatrix} = \begin{bmatrix}
          I\otimes I & -I\otimes D\\
          -D^{\top}\otimes I & I\otimes I
        \end{bmatrix}\begin{bmatrix}
          U^v(h)-U^v(0)\\
          U^v(0)-U^v(-h)
        \end{bmatrix}.
    \end{equation*}
    Since $\mathrm{det}(L)\neq 0$, we arrive at a unique solution of this system given by~\eqref{eq:G0}.\newline
    Following the same idea, the dynamic property~\eqref{eq:dynamic_property} is multiplied by $\theta^k$ and integrated to obtain
    \begin{equation*}
    \Int_{-h}^0 U'(h+\theta)\theta^k \dd \theta - \Int_{-h}^0 U'(\theta)\theta^k \dd \theta D = \Int_{-h}^0 U(h+\theta)\theta^k\dd \theta A_0 + \Int_{-h}^0 U(\theta)\theta^k \dd \theta A_1.
    \end{equation*}
    Applying integration by parts leads to
    \begin{equation*}
    -kG_{k-1} - (-h)^kU(0) + k\Bar{G}_{k-1} D + (-h)^kU(-h)D = G_k A_0 + \Bar{G}_k A_1.
    \end{equation*}
    Similar steps to the dynamic property \eqref{eq:dynamic_pro_tau_neg} give the two following algebraic equations rewritten in vectorized form
    \begin{equation*}
    \begin{bmatrix}I\otimes A_0 & I\otimes A_1\\- A_1^{\top}\otimes I & -A_0^{\top}\otimes I\end{bmatrix} \begin{bmatrix}G_k^v\\\Bar{G}_k^v\end{bmatrix} = - \begin{bmatrix}I\otimes I & -I\otimes D\\-D^{\top}\otimes I & I\otimes I \end{bmatrix} \left(k\begin{bmatrix}G_{k-1}^v\\\Bar{G}_{k-1}^v\end{bmatrix} + (-h)^k \begin{bmatrix}U^v(0)\\U^v(-h)\end{bmatrix}\right).
    \end{equation*}
    equivalently,
    \begin{equation*}
    L \begin{bmatrix}G_k^v\\\Bar{G}_k^v\end{bmatrix} = -  \left(k\begin{bmatrix}G_{k-1}^v\\\Bar{G}_{k-1}^v\end{bmatrix} + (-h)^k \begin{bmatrix}U^v(0)\\U^v(-h)\end{bmatrix}\right).
    \end{equation*}
    When $\mathrm{det}(L)\neq 0$, we conclude that the unique solution is given by~\eqref{eq:Gk}.
\end{proof}

\begin{proposition}\label{pro:H}
    If the Lyapunov condition holds and $\mathrm{det}(L)\neq 0$, then the set of matrices $\{H_{ij}\}_{i,j=\overline{0,N-1}}$ can be computed as
    \begin{equation}
        \begin{bmatrix} H_{ij}^v\\\Bar{H}_{ij}^v
        \end{bmatrix} = L^{-1}\begin{pmatrix}
      \frac{(-h)^{i+j+1}}{i+j+1}U^v(0)+(-h)^{j}G_i^v+jH_{i,j-1}^v\\
       \frac{(-h)^{i+j+1}}{i+j+1}U^v(-h)+(-h)^j\Bar{G}_i^v+j\Bar{H}_{i,j-1}^v
    \end{pmatrix}.
     \label{eq:Hk}
    \end{equation}
\end{proposition}
\begin{proof}
    Repeating the steps used in the proof of Proposition~\ref{pro:G}, the dynamic property~\eqref{eq:dynamic_property} for $\theta_1-\theta_2> 0$ gives
    \begin{equation*}
    -\frac{\partial U(\theta_1-\theta_2)}{\partial \theta_2}+\frac{\partial U(\theta_1-\theta_2-h)}{\partial \theta_2} D =  U(\theta_1-\theta_2) A_0 +  U(\theta_1-\theta_2-h) A_1 .   
    \end{equation*}
    Now, pre and post multiply the discretized dynamic property \eqref{eq:dynamic_property} by $\theta_1^i$ and $\theta_2^j$, respectively: 
\begin{equation*}
    -\theta_1^i\theta_2^j\frac{\partial U(\theta_1-\theta_2)}{\partial \theta_2} +\theta_1^i\theta_2^j\frac{\partial U(\theta_1-\theta_2-h)}{\partial \theta_2}D=  \theta_1^i\theta_2^jU(\theta_1-\theta_2) A_0 +  \theta_1^i\theta_2^jU(\theta_1-\theta_2-h) A_1 .   
\end{equation*}
Then, we integrate with respect to $\theta_2$ the previous expression, from $-h$ up to $\theta_1$ in such a way that the difference $\theta_1-\theta_2>0$. It implies that
\begin{multline*}
    -\Int_{-h}^{\theta_1}\theta_1^i
    \theta_2^j\frac{\partial U(\theta_1-\theta_2)}{\partial \theta_2}\dd \theta_2 + \Int_{-h}^{\theta_1}\theta_1^i\theta_2^j\frac{\partial U(\theta_1-\theta_2-h)}{\partial \theta_2}\dd \theta_2 D=\\  \Int_{-h}^{\theta_1}\theta_1^i\theta_2^jU(\theta_1-\theta_2)\dd \theta_2 A_0 +  \Int_{-h}^{\theta_1}\theta_1^i\theta_2^jU(\theta_1-\theta_2-h)\dd \theta_2 A_1 .   
\end{multline*}
Compute the integrals in the above expression:
\begin{equation*}
    \Int_{-h}^{\theta_1}\theta_1^i\theta_2^j\frac{\partial U(\theta_1-\theta_2)}{\partial \theta_2}\dd \theta_2= \theta_1^{i+j}U(0)-(-h)^{j}\theta_1^iU(h+\theta_1) - \Int_{-h}^{\theta_1} j\theta_1^i\theta_2^{j-1}U(\theta_1-\theta_2)\dd \theta_2,
\end{equation*}
\begin{equation*}
    \Int_{-h}^{\theta_1}\theta_1^i\theta_2^j\frac{\partial U(\theta_1-\theta_2-h)}{\partial \theta_2}\dd \theta_2 =  \theta_1^{i+j}U^{\top}(h)-(-h)^j\theta_1^iU(\theta_1) - \Int_{-h}^{\theta_1} j\theta_1^i\theta_2^{j-1}U(\theta_1-\theta_2-h) \dd \theta_2.
\end{equation*}
noting that the last term vanishes in the case $j=0$. Thus, 
\begin{multline*}
    -\theta_1^{i+j}U(0)+(-h)^{j}\theta_1^iU(h+\theta_1) + \Int_{-h}^{\theta_1} j\theta_1^i\theta_2^{j-1}U(\theta_1-\theta_2)\dd \theta_2 +\theta_1^{i+j}U^{\top}(h)D-(-h)^j\theta_1^iU(\theta_1)D\\ - \Int_{-h}^{\theta_1} j\theta_1^i\theta_2^{j-1}U(\theta_1-\theta_2-h) \dd \theta_2D=  \Int_{-h}^{\theta_1}\theta_1^i\theta_2^jU(\theta_1-\theta_2)\dd \theta_2 A_0 +  \Int_{-h}^{\theta_1}\theta_1^i\theta_2^jU(\theta_1-\theta_2-h)\dd \theta_2 A_1 .
\end{multline*}
Further, integrate with respect to $\theta_1$
\begin{multline*}
    -\Int_{-h}^0 \theta_1^{i+j}U(0)\dd \theta_1+\Int_{-h}^0(-h)^{j}\theta_1^iU(h+\theta_1)\dd \theta_1 + \Int_{-h}^0\Int_{-h}^{\theta_1} j\theta_1^i\theta_2^{j-1}U(\theta_1-\theta_2)\dd \theta_2\dd \theta_1\\ +\Int_{-h}^0\theta_1^{i+j}U^{\top}(h)D\dd \theta_1-\Int_{-h}^0(-h)^j\theta_1^iU(\theta_1)\dd \theta_1D - \Int_{-h}^0\Int_{-h}^{\theta_1} j\theta_1^i\theta_2^{j-1}U(\theta_1-\theta_2-h) \dd \theta_2\dd \theta_1D=  \\\Int_{-h}^0\Int_{-h}^{\theta_1}\theta_1^i\theta_2^jU(\theta_1-\theta_2)\dd \theta_2\dd \theta_1 A_0 +  \Int_{-h}^0\Int_{-h}^{\theta_1}\theta_1^i\theta_2^jU(\theta_1-\theta_2-h)\dd \theta_2\dd \theta_1 A_1 ,
\end{multline*}
equivalently, 
\begin{multline*}
    \frac{(-h)^{i+j+1}}{i+j+1}(U(0)-U^{\top}(h)D) + (-h)^{j}G_i +jH_{i,j-1} -(-h)^j\Bar{G}_iD-j\Bar{H}_{i,j-1}D  = H_{ij}A_0 + \Bar{H}_{ij}A_1.
\end{multline*}
Using similar steps for the dynamic property \eqref{eq:dynamic_pro_tau_neg}, we arrive at the two following algebraic equations 
\begin{multline*}
    \frac{(-h)^{i+j+1}}{i+j+1}(U(0)-U^{\top}(h)D) + (-h)^{j}G_i +jH_{i,j-1} -(-h)^j\Bar{G}_iD-j\Bar{H}_{i,j-1}D  = H_{ij}A_0 + \Bar{H}_{ij}A_1,
\end{multline*}
\begin{multline*}
    \frac{(-h)^{i+j+1}}{i+j+1}(U^{\top}(h)-D^{\top}U(0)) + (-h)^{j}\Bar{G}_i+j\Bar{H}_{i,j-1} -D^{\top}(-h)^{j}\Bar{G}_i-D^{\top}jH_{i,j-1}= - A_1^{\top}H_{ij} - A_0^{\top}\Bar{H}_{ij}.
\end{multline*}
By vectorization, the following system of linear algebraic equations is obtained
\begin{equation*}
    \begin{bmatrix}
      I\otimes A_0 & I\otimes A_1\\
      - A_1^{\top}\otimes I & -A_0^{\top}\otimes I
    \end{bmatrix}\begin{bmatrix} H_{ij}^v\\\Bar{H}_{ij}^v
        \end{bmatrix} =\begin{bmatrix}
      I\otimes I & -I\otimes D\\
      -D^{\top}\otimes I & I\otimes I
    \end{bmatrix}\begin{bmatrix}
      \frac{(-h)^{i+j+1}}{i+j+1}U(0)+(-h)^{j}G_i^v+jH_{i,j-1}^v\\
       \frac{(-h)^{i+j+1}}{i+j+1}U^{v}(-h)+(-h)^j\Bar{G}_i^v+j\Bar{H}_{i,j-1}^v
    \end{bmatrix}.
\end{equation*}
 When $\mathrm{det}(L)\neq 0$, we conclude that the unique solution is given by~\eqref{eq:Hk}.
\end{proof}
\begin{remark}
    Matrices $(G_k,\Bar{G}_k,H_{ij},\Bar{H}_{ij})$ are finally obtained by devectorization of $(G_k^v,\Bar{G}_k^v,H_{ij}^v,\Bar{H}_{ij}^v)$.
\end{remark}

With the help of Proposition~\ref{pro:G} and \ref{pro:H}, we now give an explicit expression for the integral entries of matrix $\mathbf{P}_{N}$. Thus, the integrals involving the delay Lyapunov matrix and its derivatives are given by
\begin{align*}
    J_{0N} &= U(0),\\
    J_{1N} &= \begin{bmatrix} G_0^\top A_1&\cdots&G_{N-1}^\top A_1\end{bmatrix},\\
    J_{2N} &= -\begin{bmatrix} -U^{\top}(h) & G_0^{\top} &\!\cdots\!&(N-1)G_{N-2}^{\top}\end{bmatrix}D -\begin{bmatrix} 1&\!\cdots\!&(-h)^{N-1}\end{bmatrix}U(0)D,\\
    J_{3N} &=\mathrm{He}\left(\left\{A_1^\top H_{ij}A_1\right\}_{i,j=0}^{N-1}\right),\\
    J_{4N} &=\mathrm{He}\left(\left\{A_1^\top \begin{pmatrix}\frac{(-h)^{i+j+1}}{i+j+1}U(0)+(-h)^jG_i+jH_{i,j-1}\end{pmatrix}D\right\}_{i,j=0}^{N-1}\right),\\
    J_{5N}\!+\!J_{6N} &= \mathrm{He}\left(\left\{D^\top\begin{pmatrix}-(-h)^{i+j}U(0)-i(-h)^jG_{i-1}\end{pmatrix}D\right\}_{i,j=0}^{N-1}\right)\\
    &+\mathrm{He}\left(\left\{D^\top j \begin{pmatrix}\frac{(-h)^{i+j}}{i+j}U(0)+(-h)^{j-1}G_i+(j-1)H_{i,j-2}\end{pmatrix}D\right\}_{i,j=0}^{N-1}\right).
\end{align*}

\section{Illustrative examples}\label{sec:examples}

Two examples validate the condition of Corollary~\ref{cor:nec_and_suf}. The delay Lyapunov matrix involved in each criterion is computed via the semi-analytic method introduced in \cite{kharitonov2013time} for $W=I_n$. The positive semi-definiteness of $\mathbf{P}_{N}$ is verified by using the function ‘‘cholcov" in Matlab as suggested in \cite{irina2022criterion}. In each example, the stability boundaries obtained by the D-subdivision are depicted by a solid line.  For comparison purposes, we present the order of approximations $\hat{r}$ and $N_0$ introduced in \cite{gomez2021necessary} and \cite{portilla2023}, respectively.

\textit{Example 1:} Consider the scalar neutral type system
\begin{equation*} \label{eq:time-delay_sys_scalar}
    \frac{\mathrm{d}}{\mathrm{d}t}[x(t)-dx(t-h)]=a_0 x(t)+a_1 x(t-h).
\end{equation*}
The map of tractable approximation orders for the stable region of the space of parameter $(a_0,a_1)$, shown in Figure~\ref{ex1} where the color code is related to the order $N^\star$ of $\mathbf{P}_{N^\star}$, illustrates the stability condition of Corollary~\ref{cor:nec_and_suf}. 
\begin{figure}[!t]
     \centering
      \includegraphics[width=0.7\textwidth]{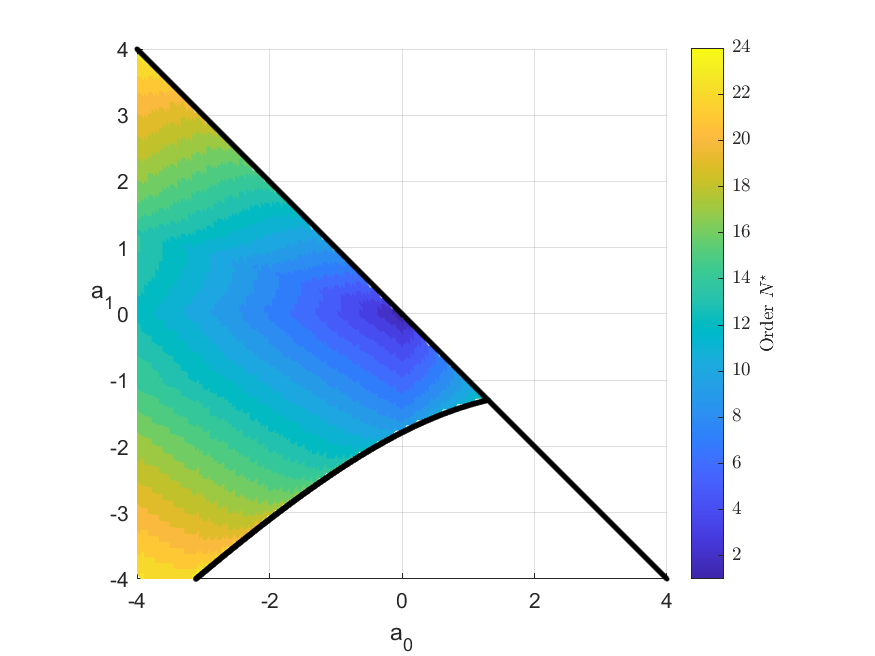}
      \caption{Map of the order $N^{\star}$ in the space of parameter $(a_0,a_1)$.}
      \label{ex1}
\end{figure}
For implementing the recursive method described in Section~\ref{sec:rec_method}, the bits of precision in Matlab play an important role. To show this, the non-negativity of $\mathbf{P}_{N}$ is verified by computing its minimum eigenvalue. In Table~\ref{table:1}, the minimum eigenvalue is computed for $8,~16$ and $32$ bits and parameter values $(a_0,a_1)=(0.8,-1.2)$, $d=-0.3$ and dimension $N=20$. It is worth mentioning that, see Figure~\ref{ex1}, system \eqref{eq:time-delay_sys_scalar} is exponentially stable for these parameters. Notice that the validation of the non-negativity of $\mathbf{P}_{N}$ fails for $8$ and $16$ bits. Thus, the bits of precision must be increased to implement the recursive method, impacting the execution time for the validation of the stability condition $\mathbf{P}_{N}$.

\begin{table}[!t]
\begin{center}
\caption{Minimum eigenvalue of $\mathbf{P}_{20}$ with respect to the precision}\label{table:1}
\begin{tabular}{ cccc } 
 \hline
  Bits of precision & $8$ bits & $16$ bits & $32$ bits\\ 
 \hline 
  $\lambda_{\min}(\mathbf{P}_{20})$ & $-4.8\time 10^{23}$ & $-2.8\time 10^{-19}$ & $1.6\time 10^{-26}$\\ 
 \hline
 \end{tabular}
\end{center}
\end{table}

\textit{Example 2:} The $\sigma$-stability analysis of the proportional-integral control of a passive linear system leads to studying a quasipolynomial of neutral type \cite{castanos2018passivity}. Its time domain representation is of the form \eqref{eq:time-delay_sys}, with matrices $D=
\begin{pmatrix}
0 & 0\\
0 & -\frac{\alpha_2}{\alpha_1}
\end{pmatrix}$,
\begin{equation*}
A_0=\frac{1}{\alpha_1}
\begin{pmatrix}
0 & \alpha_1\\
-\sigma^2\alpha_1+\sigma \beta_1 -\gamma_1 & -\beta_1+2\sigma \alpha_1 
\end{pmatrix},
\end{equation*}

\begin{equation*}
A_1=\frac{1}{\alpha_1}
\begin{pmatrix}
0 & 0\\
-\sigma^2\alpha_2+\sigma \beta_2 -\gamma_2 & -\beta_2+2\sigma \alpha_2 
\end{pmatrix},
\end{equation*}
where 
\begin{equation*}
    \begin{split}
        \alpha_1&= d+k_p,\quad \gamma_1=bk_id+ak_i,\\
        \alpha_2&= (d-k_p)e^{\sigma h},\quad \gamma_2=(bk_id-ak_i)e^{\sigma h},\\
        \beta_1&= (bk_p+a)d + bd^2 + ak_p+k_i,\\
        \beta_2&= ((bk_p+a)d-bd^2-ak_p-k_i)e^{\sigma h}.
    \end{split}
\end{equation*}

 The stability of the difference operator imposes in the D-subdivision map the additional condition $|k_p|<26.67$. For parameter values $a=0.4,~b=50,~h=0.2,~d=0.8$ and $\sigma=0.3$, the map of the approximation order $N^{\star}$ for the space of parameter $(k_p,k_i)$ is depicted in Figure~\ref{ex22}.

\begin{figure}[!t]   
     \centering
      \includegraphics[width=0.7\textwidth]{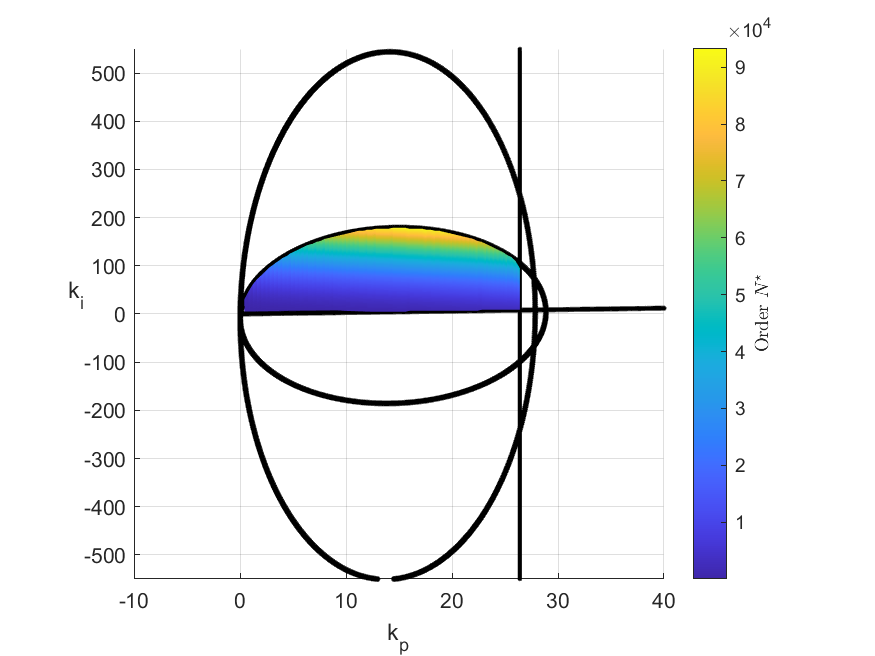}
      \caption{Map of the order $N^{\star}$ in the space of parameter $(k_p,k_i)$. }
      \label{ex22}
\end{figure}

The order of approximation $N^\star$ in Corollary~\ref{cor:nec_and_suf} is computed and compared with the results obtained in \cite{gomez2021necessary} and \cite{portilla2023} for two different points of the space of parameters. Table~\ref{table:2} shows that the value $N^\star$ is significantly smaller compared to $N_0$ and $\hat{r}$, where $(\textbf{---})$ indicates that the result is not verifiable computationally due to it exceeds the computer RAM. The favorable outcome of Corollary~\ref{cor:nec_and_suf} stems from using the Chebyshev polynomials instead of functions based on the fundamental matrix \cite{gomez2021necessary} and piece-wise linear approximation \cite{portilla2023}, resulting in a faster functional argument approximation convergence rate. Thus, we can conclude the stability of neutral-type systems in a small order of approximations.
\begin{table}[!t]
\begin{center}
\caption{Stability criteria given by Corollary~\ref{cor:nec_and_suf}, Theorem 6 in \cite{portilla2023} and Theorem 4 in \cite{gomez2021necessary}}\label{table:2}
\begin{tabular}{ ccccccc }
 \hline
 $(k_p,k_i)$ & $N^{\star}$ & Result & $N_0$ \cite{portilla2023} & Result & $\hat{r}$ \cite{gomez2021necessary} & Result\\ 
 \hline
 $(1,1)$ & $26$ & Stable & $65$ & Stable & $5\times 10^{22}$ & \textbf{---} \\ 
 $(1,-1)$ & $31$ & Unstable & $89$ & Unstable & $3\times 10^{26}$ & \textbf{---} \\ 
 \hline
 \end{tabular}
\end{center}
\end{table}

\section{Conclusions}\label{sec:conclusion}
We present necessary stability conditions for neutral time-delay systems in terms of a positive semi-definite test based on any polynomial approximation of the functional argument. The condition takes a quadratic form independent of the approximation coefficients. Its dimension depends on the polynomial approximation method. The sufficiency of the result is proven for the particular case of Chebyshev polynomial approximation. An estimation of a dimension for which the positive semi-definiteness test holds is given. Two examples show the benefits of using Chebyshev in terms of computational complexity. 

\section*{CRediT authorship contribution statement}
\textbf{Gerson Portilla:}  Methodology, Software, Validation, Writing - Review \& Editing. \textbf{Mathieu Bajodek:} Conceptualization, Methodology, Writing - Review \& Editing preparation. \textbf{Sabine Mondié:} Conceptualization, Supervision, Writing - Review \& Editing. 

\bibliography{bibliography}

\end{document}